\documentclass[letterpaper, 10 pt, conference]{ieeeconf}  

\IEEEoverridecommandlockouts                              
\makeatletter

\let\proof\@undefined
\let\endproof\@undefined

\makeatother

\usepackage{graphicx}
\usepackage{caption}
\usepackage{subfig}

\usepackage{lipsum}
\usepackage{amsfonts}
\usepackage{epstopdf}
\usepackage{caption}
\usepackage{comment}
\usepackage{amsmath,amsfonts,mathrsfs,amssymb,amscd}
\usepackage{amsthm}

\usepackage{mathtools}
\usepackage{bm}
\usepackage{lipsum}
\usepackage[linesnumbered,ruled,vlined]{algorithm2e}
\usepackage{color}
\usepackage{array}
\usepackage{multirow}
\usepackage{extpfeil}
\usepackage{cite}
\usepackage{url}
\usepackage{algorithmic}
\usepackage{lipsum,multicol}
\usepackage[ruled,vlined]{algorithm2e}

\newtheorem{theorem}{Theorem}

\newtheorem{cor}{Corollary}

\DeclareMathOperator*{\argmax}{arg\,max}

\newcommand{\mmax}{M_{\text{max}}}  
\newcommand{\isafe}{I_{\text{safeguard}}}  
\newcommand{\raa}{N_{\text{AA}}}
\newcommand{\naa}{n_{\text{AA}}}
\newcommand{\ma}{\mathcal{A}}
\newcommand{\ms}{\mathcal{S}}

\title{\LARGE \bf
On Anderson Acceleration for Partially Observable\\ Markov Decision Processes \thanks{This work was supported in part by  the Creative-Pioneering Researchers Program through SNU, the National Research Foundation of Korea funded by MSIT(2020R1C1C1009766), and the Information and Communications Technology Planning and Evaluation (IITP) grant funded by MSIT(2020-0-00857).}
}

\author{Melike Ermis
	\and
	Mingyu Park
		\and
	Insoon Yang 
	\thanks{
		M. Ermis, M. Park and I. Yang are with the Department of Electrical and Computer Engineering, Automation and Systems Research Institute,  Seoul National University, Seoul 08826, South Korea, 
		{\tt\small \{melikermis, pmg1202, insoonyang\}@snu.ac.kr}}}

\begin{document}

\maketitle
\thispagestyle{empty}
\pagestyle{empty}


\begin{abstract}
This paper proposes an accelerated method for approximately solving partially observable Markov decision process (POMDP) problems offline.
Our method carefully combines two existing tools: \emph{Anderson acceleration} (AA) and the \emph{fast informed bound} (FIB) method.
Adopting AA, our method rapidly solves an approximate Bellman equation with an efficient combination of previous solution estimates. 
Furthermore, the use of FIB alleviates the scalability issue inherent in POMDPs.
 We show the convergence of the overall algorithm to the suboptimal solution obtained by FIB.
 We further consider a simulation-based method
 and  
 prove that the approximation error is bounded explicitly. 
The performance of our  algorithm is evaluated on several benchmark  problems. 
The results of our experiments demonstrate that the proposed algorithm converges significantly faster without degrading the quality of the solution  compared to its standard counterpart. 
\end{abstract}


\section{Introduction}\label{sec:intro}
	Designing decision-making strategies  in uncertain environments is often challenging, especially under partial observability~\cite{Astrom1965}. 
Various systems have two primary sources of uncertainty: system dynamics and observations obtained via sensors~\cite{Thrun2005}. 
Thus, the successful operation of such systems is subject to 
addressing inefficiency caused by possibly uncertain movements and noisy observations. 

POMDPs have been widely recognized as a principled framework for sequential decision-making in uncertain environments~\cite{smallwood1973optimal}.
In the standard POMDP model, an agent acts in a discrete state space with discrete action and observation spaces. 
The POMDP framework takes uncertainty into account by using $(i)$ a state-transition model, and $(ii)$ a sensor model to specify noisy observations.
In POMDPs, the set of agent's possible states is represented using a probability distribution, called the \emph{belief}, over the state space. 
It is well known that the belief state is a sufficient statistic for POMDPs.

Unfortunately, it is computationally challenging to find an optimal policy of POMDPs~\cite{papadimitriou1987complexity}. The  intractability prevents POMDPs from being successfully applied in practice.
To alleviate the computational issues, several approximate POMDP algorithms (e.g.,~\cite{8594418, 8202146}) have been proposed in the literature that can be categorized as online or offline\cite{kochenderfer2015decision}. 
In the case of online planners, the best action to take is computed at the execution time. They can handle fairly high dimensional state spaces but require more computational power for real-time execution. 
In contrast, offline planners design an approximately optimal policy prior to execution. 
They cannot handle large state spaces, but they require a less online computational burden  since the execution in the environment becomes a sequence of calls to a look-up table. 
 
Some of the state-of-the-art offline algorithms, which plan in the space of belief states, such as SARSOP\cite{kurniawati2008sarsop}, point-based value iteration\cite{pineau2003point}, and  Perseus\cite{spaan2005perseus}, approximate value functions. They compute a set of ``{$\alpha$-vectors}"  to approximately represent the value function. These belief-space POMDP algorithms can provide high-performance policies. 
However, these methods are computationally demanding due to the significant increase in the number of $\alpha$-vectors. 
In contrast, state-space POMDP methods, such as Q-function approximation~\cite{littman1995learning} and FIB~\cite{hauskrecht2000value}, approximate value functions while keeping the number of $\alpha$-vectors as a constant. 
Therefore, they can alleviate the complexity and memory issues.

 The objective of this work is to develop an accelerated offline state-space  POMDP algorithm to alleviate the issue of slow convergence.
Our method uses Anderson acceleration (AA) to speed up the solution of  fixed-point problems~\cite{anderson1965iterative}. 
While the standard fixed-point iteration repeatedly applies an operator of interest to the last estimate, AA searches a linear combination of previous estimates, which minimizes the residual of the previous estimates to compute the next estimate. 
There have been attempts to apply either AA or its variants to algorithms for solving Markov decision processes (MDPs) with full observability~\cite{geist2018anderson,zhang2018globally}.  
These methods have been further extended  to deep reinforcement learning in \cite{shi2019regularized, Ermis2020}. 
It has been demonstrated on several problems that AA improves  convergence or learning speed.
However, to the best of our knowledge,  AA has not yet been applied to POMDPs.

The main results of this work can be summarized as follows:
First, we propose a novel accelerated  POMDP algorithm exploiting both AA and FIB.  To speed up convergence, we use AA with an adaptive regularization technique proposed in~\cite{fu2019anderson}. 
For scalability, we adopt FIB, which is a state-space algorithm computing one $\alpha$-vector for each action. 
To combine the two methods, we interpret FIB as a solution to a fixed-point problem.
Second, we consider a simulation-based implementation of our algorithm  and provide a provable error bound on the residual. This analysis characterizes how the corresponding approximation errors  propagated through AA  are explicitly bounded. We further show the convergence of the error-free version of our algorithm to the suboptimal solution obtained by FIB.
Third, we evaluate the performance of the policies obtained by our algorithm on several benchmark problems. 
The results of our experiments  
indicate that the proposed algorithm makes a significant improvement in the total number of iterations and the total computation time, without degrading the quality of the obtained policy.

\section{Background}

\subsection{Partially Observable Markov Decision Processes}

A partially observable Markov decision process (POMDP) is formally defined as a tuple  $(\mathcal{S}, \mathcal{A}, \mathcal{O}, T, \Omega, r, \gamma)$:
\begin{itemize}
	\item The state space $\mathcal{S}$ is a finite set of all possible states of the environment.
	\item The action space $\mathcal{A}$ is a finite set of all possible actions that the agent can take. 
	\item The observation space $\mathcal{O}$ is a finite set of all possible observations that the agent can receive. 
	\item The transition model $T$  describes the evolution of the environment over time. Given state $s \in \mathcal{S}$ and action $a \in \mathcal{A}$, 
	Specifically,
$T_{s'}^{s,a} = \mathrm{Pr} (s' | s,a)$ is
	the probability of transitioning from the current state $s$ to the next state $s'\in \mathcal{S}$ when taking action $a$.
	\item The observation model $\Omega$ describes the probability of  observing $o\in \mathcal{O}$ in a given state $s'$ after taking action $a$: $\Omega_o^{s',a}= \mathrm{Pr}(o | s',a)$.
	\item The reward function $r$ quantifies the utility of each action for each state: $r : \mathcal{S} \times \mathcal{A} \to [r_{\min}, r_{\max}]$. 
	\item The discount factor $\gamma \in (0,1)$ is a scalar that discounts the value of future rewards. 
\end{itemize}

In a POMDP, events happen in the following order: In the current state $s$, the agent takes an action $a$, and the environment transitions to a state $s'$ according to the transition model $T$. Then, the agent receives an observation $o$ related to $s'$ and $a$ according to the observation model $\Omega$. A widely used objective is to maximize the expected discounted sum of rewards over an infinite horizon, $\mathbb{E}\big[\sum_{t=0}^{\infty} \gamma^t r_t\big]$, where $r_t$ is the reward received at stage $t$.  

Although the state is not directly observed, state probabilities can be computed using Bayes' rule. Let $b(s)$ denote the probability for the environment being in state $s$.
Then, after taking action $a$ and observing $o$, 
the agent computes the probability for the environment being in state $s'$ as
\begin{equation} \label{eq:belief}
b'(s') := \tau (b, a, o) := \eta \Omega_o^{s',a} \sum_{s\in \mathcal{S}}T_{s'}^{s,a}b(s),
\end{equation}
where $\eta$ is the normalizing constant, and $b$ (or $b'$) is called the \emph{belief state} or \emph{belief}.
Since the belief state provides all information about the history  needed for an optimal solution, a POMDP can be recast as an equivalent completely observable MDP over belief states, called a \emph{belief-state} MDP, where 
\[
\mathcal{B} = \bigg \{b\in\mathbb{R}^{|\mathcal{S}|}: \sum_{s\in\mathcal{S}}b(s) = 1 \mbox{ and } b(s)\geq 0, \forall s \in \mathcal{S} \bigg \}.
\]
Using this representation, a belief-state MDP can be solved by dynamic programming (DP).
In the  infinite-horizon  discounted setting, the optimal value function $V^\star$ : $\mathcal{B} \to \mathbb{R}$
is the fixed point of the Bellman operator $\mathcal{T}$ defined as 
\begin{equation} \label{eq:exact-value}
\mathcal{T}V(b) := \max_{a\in \mathcal{A}} \bigg[ R(b,a) + \gamma\sum_{o\in\mathcal{O}} \mathrm{Pr}(o | b,a)V(\tau (b, a, o)) \bigg],
\end{equation}
where $R(b,a) := \sum_{s\in\mathcal{S}}r(s,a)b(s)$ and $\mathrm{Pr}(o | b,a) = \sum_{s\in\mathcal{S}}\sum_{s'\in\mathcal{S}}\Omega_o^{s',a}T_{s'}^{s,a}b(s)$.
The optimal value function can be found by solving the Bellman equation
\[
V = \mathcal{T} V,
\]
which is a fixed-point problem.

One standard algorithm used to find the optimal value function $V^\star$   is \emph{value iteration}. Value iteration starts with an initial guess $V_0$ and approximates $V^\star$ through repeated applications of the Bellman update $V_k  :=  \mathcal{T}V_{k-1}$. During value iteration updates, each $V_k$ is piecewise linear and convex since the Bellman operator preserves the piecewise linearity and convexity~\cite{sondik1971optimal}.
Thus, the value function can be represented using 
a finite set of $|\mathcal{S}|$-dimensional vectors
 $\Gamma_k = \{\alpha^k_1,\ldots,\alpha^k_{|\Gamma_k|}\}$ as
\[
V_k(b) := \max_{\alpha \in \Gamma_k} \sum_{s\in \mathcal{S}} \alpha (s)b(s).
\]
The standard value iteration algorithm makes an exact calculation on $\mathcal{T}$ by updating $\alpha$-vectors from the previous set $\Gamma_{k-1}$ to the current set $\Gamma_k$~\cite{sondik1971optimal,cassandra971,zhang2001speeding}. Unfortunately, in the worst case, the size of this representation grows exponentially as $|\Gamma_k| \leq |\mathcal{A}||\Gamma_{k-1}|^{|\mathcal{O}|}$, which causes computational intractability even for   small size problems. Many approximation techniques have been proposed for handling this issue in exact value iteration, as mentioned in Section~\ref{sec:intro}.  

\subsection{Fast Informed Bound Method}

State-space POMDP algorithms approximate the value function with a finite set $\Gamma_k$ of $\alpha$-vectors by associating each action with one $\alpha$-vector. 
The fast informed bound (FIB) method is a popular offline state-space POMDP algorithm,  proposed in~\cite{hauskrecht2000value}. 
Specifically,  at the $k$th iteration, FIB updates the $\alpha$-vectors  as
\begin{equation}
\label{eq:fib}
\alpha_a^{k+1}(s) := r(s,a) + \gamma\sum_{o\in\mathcal{O}}\max_{a'\in \mathcal{A}}\sum_{s'\in\mathcal{S}} \Omega_o^{s',a}T_{s'}^{s,a}\alpha_{a'}^k(s'),
\end{equation}
where $\alpha_{a} \in\mathbb{R}^{|\mathcal{S}|}$ for all $a \in \mathcal{A}$, instead of using the exact value iteration algorithm.
Note that 
  the size of the set $\Gamma_k := \{\alpha^k_1,\ldots,\alpha^k_{|\mathcal{A}|}\}$ remains the same throughout the update process, thereby alleviating the scalability issue in POMDPs.
It is well-known that the FIB update upper-bounds the exact value iteration update and is tighter than other popular value-function approximation techniques, namely the MDP and QMDP approximation approaches~\cite{Astrom1965, Lovejoy1993, littman1995learning}.
In fact, the FIB update~\eqref{eq:fib} can be derived from the exact Bellman update~\eqref{eq:exact-value} by decoupling the expectation over $b(s)$ from $R(b,a)$ and $\mathrm{Pr}(o | b,a)$.
Thus, the FIB update becomes independent of the belief, similar to the MDP and QMDP approximation.
However, FIB takes uncertainty from observations into account, and thus intuitively the FIB solution is closer to the optimal solution than the approximate solutions obtained by the other methods.

\subsection{Anderson Acceleration}
 
 Consider a function $F: \mathbb{R}^N \to \mathbb{R}^N$, where the associated fixed-point problem can be defined as 
\[
x = F(x).
\]
The fixed-point iteration (FPI) repetitively applies 
\[
x^{k+1} = F(x^k).
\]
These repetitions make FPI algorithms converge slowly. 

To speed up the computation of a fixed point,
 Anderson acceleration (AA) can be used instead of the naive FPIs~\cite{anderson1965iterative}. 
 At the $k$th iteration,
AA  maintains the most recent $M^k + 1$ estimates $(x^k, \ldots, x^{k-M^k})$ in memory. Here, $M^k \in \mathbb{Z}_+$ is the memory size and is given by $M^k =\min(M_{\max}, k)$  with some hyper-parameter $M_{\max}$~\cite{walker2011anderson}. 
Then, the estimate of $x$ is updated as the following weighted sum of $F(x^{k - M^k + i})$'s:
\begin{equation} \label{eq:anderson-update}
x^{k+1} = \sum_{i=0}^{M^k} w_i^k F( x^{k - M^k + i}).
\end{equation}  
Here, the weight vector $w^k = (w_0^k, \ldots, w_{M_k}^{k}) \in \mathbb{R}^{M^k + 1}$ is obtained by solving the following optimization problem:
\begin{equation}\label{eq:anderson-problem}
\begin{split}
\min_{w^k} \quad  &\bigg \| \sum_{i=0}^{M^k} w_i^k G (x^{k - M^k + i}) \bigg \|_2^2\\
\mbox{s.t.}  \quad  &\sum_{i=0}^{M^k} w_i^k =1,
\end{split}
\end{equation}
where
\[
G (x) := x - F(x)
\]
represents the residual. 
Then, the memory is updated to $(x^{k+1}, \ldots, x^{k+1-M^{k+1}})$ before repeating the process. 
The update \eqref{eq:anderson-update} can be interpreted as an extrapolation mechanism to rapidly reduce the residual by using the previous iterates in the memory. 
Several variants of AA have also been studied in the literature~\cite{scieur2016regularized, zhang2018globally, fu2019anderson}.

\section{Accelerated POMDP Algorithm}

 Our goal is to devise an accelerated POMDP algorithm that 
 carefully combines AA and FIB. 
 The key idea of our approach is two-fold: 
 The first is to update the set $\Gamma_{k+1}$ of $\alpha$-vectors 
  not only relying on $\Gamma_{k}$  but also using the previous sets $\Gamma_k, \ldots, \Gamma_{k - M^k}$ in memory.
Second, to exploit AA in finding an efficient linear combination of the previous estimates, we represent the FIB algorithm as a solution to a fixed-point problem. 
The use of FIB enhances the scalability of our method, while AA helps in rapidly finding the fixed point. 

\subsection{Anderson Acceleration for POMDPs}\label{sec:A3FIB}

We first define the fixed-point problem to apply AA in FIB, where the size of the set $\Gamma = \{\alpha_1,\ldots,\alpha_{|\mathcal{A}|}\}$ remains the same throughout the $\alpha$-vector update. 
Let the vectorization of the set $\Gamma$ be defined as
\[
\alpha := \begin{bmatrix}
\alpha_1 \\
\alpha_2 \\
\vdots \\
\alpha_{|\ma|}
\end{bmatrix} \in  \mathbb{R}^{|\ms|   |\ma|},
\]
and  let
$F: \mathbb{R}^{|\ms|   |\ma|} \to \mathbb{R}^{|\ms|   |\ma|}$ denote the approximate Bellman operator for updating the $\alpha$-vector as \eqref{eq:fib}.
More precisely, 
\begin{equation}\label{app_bellman}
\begin{split}
&(F\alpha) ((a-1)|\ms| + s) \\
&:= r(s,a) + \gamma\sum_{o\in\mathcal{O}}\max_{a'\in \mathcal{A}}\sum_{s'\in\mathcal{S}} \Omega_o^{s',a}T_{s'}^{s,a}\alpha_{a'} (s').
\end{split}
\end{equation}
It is straightforward to check that $F$ is a $\gamma$-contraction mapping with respect to $\| \cdot \|_\infty$ since $\mathcal{S}$, $\mathcal{A}$ and $\mathcal{O}$ are finite sets. 
Thus, by the Banach fixed-point theorem, the standard FIB algorithm
\begin{equation}\nonumber
\alpha^{k+1} := F\alpha^k
\end{equation}
converges to the fixed point of $F$.
The corresponding fixed-point  problem $\alpha = F \alpha$ can be regarded as an approximate Bellman equation.

To accelerate the process of finding the fixed point, 
we use AA on the approximate Bellman operator $F$.
 Let the residual function be  defined as
\[
G(\alpha) := \alpha - F\alpha \in \mathbb{R}^{|\mathcal{S}| |\mathcal{A}|}.
\]
Suppose that the estimates of $\alpha$ have been computed up to the $k$th iteration and  consider the most recent $M^k +1$ estimates $(\alpha^k, \ldots, \alpha^{k-M^k} )$. 
Then, our accelerated version of the FIB algorithm updates the next estimate as 
\begin{equation}\label{eq:raa-q}
\alpha^{k+1} := \sum_{j=0}^{M^k} w_j^k F \alpha^{k - M^k + j},
\end{equation} 
where the weight vector	$w^k = (w^k_0,\ldots,w^k_{M^k})$ is obtained through the AA algorithm that will be introduced in what follows. 

Let
\[
g^k := G (\alpha^k), \quad y^k := g^{k+1} - g^k, \quad s^k := \alpha^{k+1} - \alpha^k,
\]
and
\begin{equation}\label{ys}
Y_k := \begin{bmatrix}
y^{k - M^k} & \cdots & y^{k-1}
\end{bmatrix}, \: 
S_k := \begin{bmatrix}
s^{k - M^k} & \cdots & s^{k-1}
\end{bmatrix}.
\end{equation}
Note that $Y_k$ and $S_k$ are $|\mathcal{S}| |\mathcal{A}|$ by $M^k$ matrices.
We introduce a new variable $\xi^k = (\xi^k_0, \dots , \xi^k_{M^k-1})$ that is related to the original weight vector $w^k$ as follows:
\begin{equation} \label{eq:relation}
\begin{split}
w_0^k &:= \xi_0^k, \quad w_i^k := \xi_i^k - \xi_{i-1}^k, \quad 1 \leq i \leq M^k - 1   \\
w_{M^k}^k &:= 1- \xi_{M^k -1}^k.
\end{split}
\end{equation} 
Then, the optimization problem~\eqref{eq:anderson-problem} can be reformulated as the following least-squares problem:
\[
\min_{\xi^k} \; \| g^k - Y_k \xi^k \|_2^2.
\]
To alleviate the instability issue in the original AA algorithm, 
we use an additional $\ell_2$-regularization term scaled by the Frobenius norms of $S_k$ and $Y_k$, as proposed in~\cite{fu2019anderson}. 
Specifically, we consider the following modified optimization problem to compute the weights:
\begin{equation}\label{eq:a3}
\min_{\xi^k} \; \| g^k - Y_k  \xi^k \|_2^2 + \eta  ( 
	\| S_k \|_F^2 + \|Y_k \|_F^2
	) 
	\| \xi^k \|_2^2,
\end{equation}
where $\eta > 0$ is a regularization parameter.
Note that the regularization term is adaptive to $y^k$ and $s^k$ that are proxies for assessing how close the current estimate is to the fixed point. 
This regularized least-squares problem admits a closed-form optimal solution, given by
\begin{equation}
\label{eq:sol_xi}
\xi^k = (Y_k^\top Y_k + \eta ( \| S_k \|_F^2 + \|Y_k \|_F^2) I )^{-1} Y_k^\top g^k.
\end{equation}
The weight vector $w^k$ is then recovered by using the relation between $w^k$ and $\xi^k$ in \eqref{eq:relation}. After finding the weight vector, our accelerated FIB algorithm computes the next estimate $\alpha^{k+1}$ by substituting $w^k$ into \eqref{eq:raa-q}.

The overall algorithm, which we call AA-FIB, is presented in Algorithm~\ref{A3-FIB}.
The algorithm starts with selecting all the components of $\alpha^0 \in \mathbb{R}^{|\mathcal{S}||\mathcal{A}|}$ randomly from $[{r_{\min}}/{(1-\gamma)},{r_{\max}}/{(1-\gamma)}]$.
At iteration $k$, the FPI candidate $\alpha_{\text{FPI}}$ and the AA candidate $\alpha_{\text{AA}}$ are computed (lines 6--10). 
Then, the algorithm decides to select either the AA candidate or the FPI candidate in the \emph{safeguarding} step, as proposed in~\cite{fu2019anderson} (lines 13--21). 
It has been shown that similar ideas of safeguarding may improve the convergence and stability of AA \cite{zhang2018globally}.
The safeguarding step checks whether the current residual norm is sufficiently small. If true, the algorithm uses the AA estimate $\alpha_{\text{AA}}$ and the next $N_s-1$ iterations skip the safeguarding step. Otherwise, instead of using the AA estimate, the FPI estimate $\alpha_{\text{FPI}}$ is taken. 
This process is repeated until convergence, and the set $\Gamma$ of $\alpha$-vectors can be recovered from the converged $\alpha$.
Finally, the corresponding policy can be obtained as
\[
\pi(b) \in \argmax_{a \in \mathcal{A}} \, \{b^\top  \alpha_a  : \alpha_a \in \Gamma \}.
\]
In Section~\ref{sec:theory}, we will prove the convergence of the algorithm to the FIB solution $\alpha^\star$, which is a suboptimal solution of the POMDP.
This demonstrates that the use of AA does not degrade the quality of the solution.

\begin{algorithm}
	\begin{algorithmic}[1]
		\caption{AA-FIB}\label{A3-FIB}
		\STATE \textbf{Input}: Initial vector $\alpha^{0}$, memory size $\mmax$, regularization parameter $\eta$, safeguarding constants $D > 0$,  $\phi > 0$, $N_{s} \geq 1$;\\
		\STATE Initialize $\isafe :=$ True, $\naa := 0$, $\raa := 0$; \\
		\STATE Compute  $\alpha^{1} := F \alpha^{0}$, $g^{0} := G (\alpha^0)$;\\
 		\FOR{$k = 1,2, \dots $}
		\STATE Set $M^k := \min(\mmax, k)$; 
		\STATE Compute the FPI candidate and residual: $\alpha_{\text{FPI}} := F  \alpha^{k}$, 
		$g^{k} := G (\alpha^k)$;   
		\STATE \text{\# \bf Adaptive regularization}
		\STATE Solve the regularized least squares problem \eqref{eq:a3} or compute the closed-form solution~\eqref{eq:sol_xi};
		\STATE Compute $w^{k}$ from \eqref{eq:relation};
		\STATE Compute the AA candidate as $\alpha_{\text{AA}} := \sum_{i=0}^{M^k} w_{i}^{k}F  \alpha^{k-M^k+i}$;
		\STATE \text{\# \bf Safeguarding}
		\IF{$\isafe = \text{True}$ or $\raa \geq N_s$}
		\IF{$\|g^{k}\|_\infty \leq D\|g^{0}\|_\infty(\naa /N_s+1)^{-(1+\phi)}$}
		\STATE $\alpha^{k+1} := \alpha_{\text{AA}} $, $\naa := \naa+1$, \\$\isafe := \text{False}$, $\raa := 1$ ;
		\ELSE
		\STATE $\alpha^{k+1} := \alpha_{\text{FPI}}$, $\raa := 0$;
		\ENDIF
		\ELSE 
		\STATE $\alpha^{k+1} := \alpha_{\text{AA}}$, $\naa := \naa+1$,
		$\raa := \raa + 1$; 
		\ENDIF 
		\STATE Terminate if the stopping criterion is satisfied;
		\ENDFOR
		\STATE \textbf{return} $\alpha^{k+1}$;
	\end{algorithmic}
\end{algorithm}

\subsection{Simulation Method}\label{sec:app}

So far, we have considered the case where the POMDP model is fully known. 
However, when the model information is unavailable, i.e., $T$, $\Omega$, $r$ are unknown, it is impossible to  implement Algorithm~\ref{A3-FIB}.
In such cases, simulation-based methods can be used to obtain an approximate solution when a POMDP simulator is available~(e.g., \cite{Bertsekas2012}).

We adopt the setting in \cite{Silver2010}, where
the agent uses a simulator $\mathcal{G}$ as a \emph{generative} model of the POMDP. 
The simulator generates a sample of a successive state, observation and reward $\{(s_j', o_j, r_j)\}_{j \in J}$, given the current state and action pair, i.e., $(s_j', o_j, r_j) \sim \mathcal{G} (s, a)$.\footnote{Note that $s_j'$ is different from the actual next state, which is unobservable, since $s_j'$ is produced from the generative model. However, its sample has information about the transition probability.}
We consider the following simulation-based version of the Bellman operator~\eqref{app_bellman}:
\begin{equation}\label{sim_bellman}
\begin{split}
&(\hat{F}\alpha) ((a-1)|\ms| + s) \\
&:= \frac{1}{| J|} \bigg [ \sum_{j \in J} r_j + \gamma  \sum_{o \in \mathcal{O}} \max_{a'\in \mathcal{A}}\sum_{j \in J} \hat{\Omega}_{o}^{s_j', a}  \alpha_{a'} (s_j')\bigg ].
\end{split}
\end{equation}
Here, $\hat{\Omega}_{o}^{s_j', a}$ denotes the empirical distribution of observation given $(s_j', a)$, defined as
\[
\hat{\Omega}_{o}^{s', a} := \frac{1}{|J_{s'}|}\sum_{j \in J_{s'}} \bold{1}_{\{o_j = o\}},
\]
where $J_{s'}$ denotes the set of indices such that $s_j = s'$ for any $j \in J_{s'}$, and $\bold{1}_A$ is the indicator function with $\bold{1}_A = 1$ if $A$ is true and $\bold{1}_A = 0$ otherwise. 
Algorithm~\ref{A3-FIB} is then accordingly modified by replacing $F$ with $\hat{F}$, and $G$ with $\hat{G}$ defined as
\[
\hat{G}(\alpha) := \alpha - \hat{F} \alpha.
\]
Unlike the original AA-FIB algorithm, this simulation-based version is not guaranteed to converge. 
Nonetheless, we provide an error bound in the following subsection.

\subsection{Error Bound and Convergence}\label{sec:theory}
 
We now examine how the simulation errors  are propagated through the AA-FIB algorithm. 
Consider the following mismatch between $\hat{F}\alpha^{k}$ and $F \alpha^k$:
\[
e^k := \hat{F}\alpha^{k} - F \alpha^k,
\]
which can be interpreted as the error caused by the simulation.
Then, the FPI step in the simulation-based AA-FIB algorithm can be expressed as
\[
\alpha^{k+1} =  \hat{F}  \alpha^k = F  \alpha^k   + e^k.
\]
We show that the following error bound holds in terms of the residual function $G (\alpha^k) = \alpha^k - F \alpha^k$.

\begin{theorem}\label{thm:app}
Consider the simulation-based AA-FIB algorithm.
	Suppose there exists $\varepsilon \geq 0$ such that  $\|e^k\|_\infty \leq \varepsilon$ for all $k\geq 0$. 
	Then, we have
			\begin{equation} \nonumber
	\lim_{k\rightarrow \infty}\inf \|G(\alpha^k)\|_\infty \leq \frac{1+\gamma}{1-\gamma}\varepsilon.	
	\end{equation}
\end{theorem}

\begin{proof}
We first consider the case in which AA candidates are selected for finitely many times in Algorithm~\ref{A3-FIB}.	 
In this case, 
	 the algorithm reduces to the vanilla value iteration process after a finite number of iterations less than a bound, say $k$. 
Let $\alpha^\star$ denote the FIB solution, i.e., the fixed-point of $F$.
It follows from the $\gamma$-contractivity of $F$ that 
	\[
	\|F \alpha^{k} - \alpha^\star\|_\infty = 	\|F \alpha^{k} - F \alpha^\star \|_\infty \leq \gamma \|\alpha^{k}-\alpha^\star\|_\infty.
	\]
Then, we have	
\begin{equation}\label{ineq1}
\begin{split}
\| G(\alpha^{k}) \|_\infty &= \| \alpha^{k} - F \alpha^{k}  \|_\infty\\
&= \| \alpha^{k} - \alpha^\star + \alpha^\star   - F \alpha^{k} \|_\infty \\
&\leq (1+\gamma) \|\alpha^{k} - \alpha^\star \|_\infty.
\end{split}
\end{equation}	
Moreover, since	
	\begin{equation}\nonumber
	\|\hat{F}  \alpha^k  - F \alpha^k  \|_\infty = \|e^k\|_\infty \leq \varepsilon,
	\end{equation}	
we obtain
	\begin{equation} \label{eq:cont}
	\begin{split}
	\|\alpha^{k+1} - \alpha^{\star}\|_\infty 
	&\leq 
	\|\hat{F}\alpha^{k} - F\alpha^{k}\|_\infty + \|F\alpha^{k} - \alpha^\star\|_\infty
\\
	&\leq \varepsilon +   \gamma  \|\alpha^{k}-\alpha^\star\|_\infty.
	\end{split}
	\end{equation} 	
 Taking $\liminf$ on both sides yields
 \begin{equation}\label{ineq2}
 \liminf_{k \to \infty} \|\alpha^{k} - \alpha^{\star}\|_\infty 
 \leq \frac{\varepsilon}{1-\gamma}.
 \end{equation}
It follows from \eqref{ineq1} and \eqref{ineq2} that 
 \[
 \liminf_{k \to \infty} \| G(\alpha^{k}) \|_\infty \leq \frac{1+\gamma}{1-\gamma}\varepsilon.
 \]

Suppose now that AA candidates are selected for infinitely many times in Algorithm~\ref{A3-FIB}.	
Recall that
	${g}^k := \hat{G} (\alpha^k )$.
	Since $\|e^k\|_\infty \leq \varepsilon$,
	\begin{equation} \nonumber
	\begin{split}
	\| g^k - G(\alpha^k)\|_\infty &= \|\hat{G}(\alpha^k) - G(\alpha^k)\|_\infty \\
	&= \|\hat{F} \alpha^k - F\alpha^k\|_\infty \\
	&\leq \|e^k \|_\infty \leq \varepsilon.
	\end{split}
	\end{equation}
Let $k_i$ denote the initial iteration count for accepting an AA candidate.
 Then, the set of $k_i$'s has infinitely many elements, and we have
	\begin{equation} \label{eq:G}
	\begin{split}
	&\lim_{k\rightarrow \infty}\inf \|G(\alpha^k)\|_\infty\\
	 &\leq  \lim_{i\rightarrow \infty}\inf \|G(\alpha^{k_i})\|_\infty\\
	& \leq \lim_{i\rightarrow \infty}\inf (\|g^{k_i}\|_\infty + \| g^k - G(\alpha^k)\|_\infty ) \\
	& \leq \lim_{i\rightarrow \infty}\inf \|g^{k_i}\|_\infty + \varepsilon \\
	&\leq D\|g^0\|_\infty \lim_{i\rightarrow\infty}(i+1)^{-(1+\phi)} + \varepsilon\\
	& \leq \varepsilon \leq \frac{1+\gamma}{1-\gamma}\varepsilon,
	\end{split}
	\end{equation}
	where  we use the fact that $\|g^{k_i}\|_\infty \leq D\|g^{0}\|_\infty (i+1)^{-(1+\phi)}$ in the fourth inequality.
	\end{proof}

Note that the error bound in Theorem~\ref{thm:app} is linear in $\varepsilon$ and 
depends only on $\varepsilon$ and $\gamma$. 
This implies that AA has no explicit impact on the error bound since $\varepsilon$ is the error caused by simulations and $\gamma$ is the discount factor. 
From this observation, we  deduce the following convergence of the exact AA-FIB algorithm:
	
\begin{cor}
Consider the original AA-FIB algorithm with $e^k \equiv 0$ for all $k \geq 0$, assuming that the model information is available.  Then, we have 
\begin{equation}\nonumber
\alpha^k \to \alpha^\star \quad \mbox{as $k \to \infty$},
	\end{equation}
	where $\alpha^\star$ is the FIB solution, i.e., the fixed-point of $F$. 
\end{cor}	
\begin{proof}
First, when AA candidates are selected for finitely many times, it follows from~\eqref{eq:cont} and $e^k = 0$ that
\[
\|\alpha^{k+1} - \alpha^{\star}\|_\infty \leq \gamma \|\alpha^{k} - \alpha^{\star}\|_\infty \quad \forall k \geq k'.
\]
This directly implies that $\alpha^{k+1} \to \alpha^\star$ as $k \to \infty$. 

Second, suppose that AA candidates are selected for infinitely many times, and let $k_i$ denote the initial iteration count for accepting an AA candidate.
Note that 
\begin{equation}\nonumber
\begin{split}
\| G(\alpha^k) \|_\infty &= \| \alpha^k - F\alpha^k \|_\infty \\
&=\| (\alpha^k - \alpha^\star) - (F \alpha^k - F\alpha^\star) \|_\infty \\
&\geq \|\alpha^k - \alpha^\star \| - \| F \alpha^k - F\alpha^\star \|_\infty\\
&\geq (1-\gamma) \|\alpha^k - \alpha^\star \|_\infty.
\end{split}
\end{equation}
Then, using the same logic as in \eqref{eq:G} with $e^k = 0$, we have
\[
(1-\gamma) \|\alpha^{k_i} - \alpha^\star \|_\infty \leq \| G(\alpha^{k_i}) \|_\infty \leq D \| g^0 \|_\infty  (i+1)^{-(1+\phi)}.
\]
For any $k$, there exists $i$ such that $k_i \leq k < k_{i+1}$ and we have
\begin{equation} \nonumber
\begin{split}
\| \alpha^k - \alpha^\star \|_\infty &\leq (1-\gamma)^{k - k_i} \| \alpha^{k_i} - \alpha^\star\|_\infty\\
& \leq (1-\gamma)^{k - k_i-1}  D \| g^0 \|_\infty  (i+1)^{-(1+\phi)}.
\end{split}
\end{equation}
As $k \to \infty$, we have $i \to \infty$, and thus $\alpha^k \to \alpha^\star$.
\end{proof}

The result of convergence to the FIB solution implies that AA does not degrade the quality of the solution. 
This observation is consistent with the results of our numerical experiments in  the following section. 

\section{Numerical Experiments}

 \begin{table*}[tb] 
	\caption{Performance comparisons for the exact AA-FIB algorithm on benchmark problems ($\mathrm{mean} \pm \mathrm{std}$).}
	\begin{center}
		\scalebox{0.89}{
			\begin{tabular}{|l|l|c|c|c|c|c|c|}
				\hline
				& & \multicolumn{6}{c|}{\textbf{Algorithm}}  \\  \hline
				\multirow{2}{*}{\vtop{\hbox{\strut \textbf{Problem}}\hbox{\strut \textbf{($|\mathcal{S}|, |\mathcal{A}|, |\mathcal{O}|$)}}}} & \multirow{2}{*}{\textbf{Metric}}  & \multirow{2}{*}{\textbf{FIB}} &  \multicolumn{4}{c|}{\textbf{AA-FIB (memory size $M_{\max}$)}} & \multirow{2}{*}{\textbf{SARSOP}}\\ 
				& & & $M_{\max}=4$ & $M_{\max}=8$ & $M_{\max}=12$ & $M_{\max}=16$ &\\\hline \hline
				
				\multirow{4}{*}{\vtop{\hbox{\strut cit}\hbox{\strut ($284,4,28$)}}} 
		    	&  $\#\mathrm{iter}$        & $1362.01\pm 17.78$     & $\bf{507.77 \pm 98.69}$      & $645.49 \pm 151.52$     & $614.83 \pm 145.75$     & $524.71 \pm 152.88$     & -\\ 
				&  $t_{\mathrm{total}}$(sec)   & $4.036 \pm 0.057$       & $\bf{1.813 \pm 0.351}$       & $2.372 \pm 0.555$       & $2.241 \pm 0.529$       & $1.954 \pm 0.568$       & $1.813$ \\ 
				&  $t_{\mathrm{AA}}$(sec)      & -                     & $0.047 \pm 0.010$       & $0.074 \pm 0.018$       & $0.074 \pm 0.018$       & $0.069 \pm 0.020$       & -\\ 
				&  $\mathrm{reward}_{\mathrm{rand}}$     & $0.44 \pm 0.05$         & $0.44 \pm 0.05$         & $0.44 \pm 0.05$         & $0.43 \pm 0.05$         & $0.44 \pm 0.05$         & $0.03 \pm 0.01$\\
				&  $\mathrm{reward}_{\mathrm{fixed}}$    & $0.81 \pm 0.01$         & $0.81 \pm 0.01$         & $0.81 \pm 0.01$         & $0.81 \pm 0.01$         & $0.81 \pm 0.01$         & $0.83 \pm 0.01$\\ \hline
				
				\multirow{4}{*}{\vtop{\hbox{\strut mit}\hbox{\strut ($204,4,28$)}}}
				&  $\#\mathrm{iter}$        & $1362.24 \pm 12.39$     & $391.84 \pm 83.56$      & $454.24 \pm 111.57$     & $396.68 \pm 107.60$     & $\bf{323.47 \pm 99.73}$      & -\\ 
				&  $t_{\mathrm{total}}$(sec)   & $4.021 \pm 0.054$       & $1.303 \pm 0.278$       & $1.583 \pm 0.390$       & $1.392 \pm 0.379$       & $\bf{1.144 \pm 0.355}$       & $1.144$ \\ 
				&  $t_{\mathrm{AA}}$(sec)      & -                     & $0.035 \pm 0.008$       & $0.050 \pm 0.012$       & $0.048 \pm 0.013$       & $0.040 \pm 0.012$       & -\\ 
				&  $\mathrm{reward}_{\mathrm{rand}}$     & $0.65 \pm 0.04$         & $0.64 \pm 0.04$         & $0.64 \pm 0.03$         & $0.65 \pm 0.34$         & $0.64 \pm 0.04$         & $0.14 \pm 0.04$\\
				& $\mathrm{reward}_{\mathrm{fixed}}$    & $0.86 \pm 0.02$         & $0.85 \pm 0.01$         & $0.86 \pm 0.01$         & $0.86 \pm 0.01$         & $0.86 \pm 0.12$         & $0.86 \pm 0.01$\\ \hline
				
				\multirow{4}{*}{\vtop{\hbox{\strut pentagon}\hbox{\strut ($212,4,28$)}}}
				&  $\#\mathrm{iter}$        & $1361.22 \pm 17.87$     & $354.39 \pm 83.82$      & $455.26 \pm 116.38$     & $412.32 \pm 125.75$     & $\bf{337.46 \pm 122.15}$     & -\\ 
				&  $t_{\mathrm{total}}$(sec)   & $3.797 \pm 0.052$       & $1.155 \pm 0.272$       & $1.495 \pm 0.381$       & $1.367 \pm 0.420$       & $\bf{1.122 \pm 0.408}$       & $1.122$ \\ 
				&  $t_{\mathrm{AA}}$(sec)      & -                     & $0.030 \pm 0.008$       & $0.047 \pm 0.012$       & $0.047 \pm 0.015$       & $0.040 \pm 0.015$       & -\\ 
				&  $\mathrm{reward}_{\mathrm{rand}}$     & $0.57 \pm 0.04$         & $0.58 \pm 0.05$         & $0.57 \pm 0.04$         & $0.58 \pm 0.04$         & $0.58 \pm 0.05$         & $0.05 \pm 0.02$\\
				&  $\mathrm{reward}_{\mathrm{fixed}}$    & $0.83 \pm 0.01$         & $0.83 \pm 0.02$         & $0.83 \pm 0.02$         & $0.83 \pm 0.02$         & $0.83 \pm 0.02$         & $0.77 \pm 0.02$\\ \hline
				
				\multirow{4}{*}{\vtop{\hbox{\strut sunysb}\hbox{\strut ($300,4,28$)}}}
				&  $\#\mathrm{iter}$        & $1364.77 \pm 11.13$     & $\bf{472.05 \pm 95.96}$      & $685.37 \pm 150.54$     & $616.37 \pm 144.80$     & $521.82 \pm 162.25$     & -\\ 
				&  $t_{\mathrm{total}}$(sec)   & $4.280 \pm 0.071$       & $\bf{1.676 \pm 0.341}$       & $2.465 \pm 0.540$       & $2.230 \pm 0.543$       & $1.899 \pm 0.592$       & $1.676$\\ 
				&  $t_{\mathrm{AA}}$(sec)      & -                     & $0.045 \pm 0.010$       & $0.079 \pm 0.018$       & $0.072 \pm 0.234$       & $0.069 \pm 0.022$       & -\\ 
				&  $\mathrm{reward}_{\mathrm{rand}}$     & $0.42 \pm 0.04$         & $0.43 \pm 0.04$         & $0.43 \pm 0.05$         & $0.43 \pm 0.04$         & $0.43 \pm 0.05$         & $0.09 \pm 0.03$\\
				&  $\mathrm{reward}_{\mathrm{fixed}}$   & $0.76 \pm 0.02$         & $0.76 \pm 0.02$         & $0.76 \pm 0.02$         & $0.76 \pm 0.02$         & $0.76 \pm 0.02$         & $0.72 \pm 0.02$\\ \hline
				
				\multirow{5}{*}{\vtop{\hbox{\strut fourth}\hbox{\strut ($1052,4,28$)}}}
				&  $\#\mathrm{iter}$        & $1364.02 \pm 14.57$     & $\bf{672.01 \pm 46.82}$      & $1002.76 \pm 129.27$    & $1142.27 \pm 9.68$      & $1213.37 \pm 327.84$    & -\\ 
				&  $t_{\mathrm{total}}$(sec)   & $6.241 \pm 0.092$       & $\bf{3.419 \pm 0.237}$       & $5.164 \pm 0.663$       & $5.960 \pm 1.003$       & $6.396 \pm 1.747$       & $3.419$\\ 
				&  $t_{\mathrm{AA}}$(sec)     & -                     & $0.084 \pm 0.006$       & $0.147 \pm 0.019$       & $0.199 \pm 0.034$       & $0.243 \pm 0.067$       & -\\ 
				& $\mathrm{reward}_{\mathrm{rand}}$    & $0.39 \pm 0.03$         & $0.40 \pm 0.03$         & $0.40 \pm 0.03$         & $0.39 \pm 0.03$         & $0.40 \pm 0.03$         & $0.06 \pm 0.02$\\
				&  $\mathrm{reward}_{\mathrm{fixed}}$    & $0.59 \pm 0.01$         & $0.59 \pm 0.01$         & $0.59 \pm 0.01$         & $0.59 \pm 0.01$         & $0.59 \pm 0.01$         & $0.50 \pm 0.02$\\ \hline
				
				\multirow{4}{*}{\vtop{\hbox{\strut Tag}\hbox{\strut ($870,5,30$)}}}
				&  $\#\mathrm{iter}$        & $315.52 \pm 0.52$       & $100.12 \pm 8.65$       & $95.93 \pm 11.52$       & $86.26 \pm 9.68$        & $\bf{85.03 \pm 8.29}$        & -\\ 
				&  $t_{\mathrm{total}}$(sec)   & $1.744 \pm 0.011$       & $0.575 \pm 0.051$       & $0.571 \pm 0.072$       & $0.518 \pm 0.063$       & $\bf{0.514 \pm 0.054}$       & $0.514$\\ 
				&  $t_{\mathrm{AA}}$(sec)     & -                     & $0.012 \pm 0.001$       & $0.014 \pm 0.002$       & $0.014 \pm 0.002$       & $0.016 \pm 0.002$       & -\\ 
				&  $\mathrm{reward}_{\mathrm{rand}}$    & $-15.89 \pm 0.88$       & $-15.87 \pm 0.93$       & $-15.85 \pm 0.99$       & $-15.88 \pm 0.94$       & $-15.89 \pm 0.95$       & $-5.87 \pm 0.57$\\
				&  $\mathrm{reward}_{\mathrm{fixed}}$    & $-17.35 \pm 0.70$       & $-17.36 \pm 0.77$       & $-17.32 \pm 0.71$       & $-17.31 \pm 0.71$       & $-17.28 \pm 0.63$       & $-6.12 \pm 0.66$\\ \hline
				
				\multirow{4}{*}{\vtop{\hbox{\strut Underwater}\hbox{\strut ($2653,6,102$)}}}
				&  $\#\mathrm{iter}$        & $445.96 \pm 4.96$       & $\bf{145.40 \pm 5.72}$       & $182.77 \pm 7.11$       & $212.07 \pm 11.11$      & $234.88 \pm 13.40$      & -\\ 
				&  $t_{\mathrm{total}}$(sec)   & $16.965 \pm 0.344$      & $\bf{5.654 \pm 0.242}$       & $7.117 \pm 0.291$       & $8.286 \pm 0.447$       & $9.286 \pm 0.545$       & $5.654$\\ 
				&  $t_{\mathrm{AA}}$(sec)      & -                     & $0.033 \pm 0.004$       & $0.057 \pm 0.007$       & $0.088 \pm 0.009$       & $0.132 \pm 0.015$       & -\\ 
			    &  $\mathrm{reward}_{\mathrm{rand}}$     & $3298.38 \pm 246.15$    & $3329.90 \pm 262.39$    & $3277.46 \pm 258.17$    & $3261.39 \pm 230.99$    & $3300.93 \pm 240.30$    & $3301.69 \pm 266.82$\\
				&  $\mathrm{reward}_{\mathrm{fixed}}$    & $-39.86 \pm 43.86$      & $-37.88 \pm 47.32$      & $-50.83 \pm 44.61$      & $-48.07 \pm 50.10$      & $-47.36 \pm 41.97$      & $706.76 \pm 7.81$\\ \hline
		\end{tabular}
		}
		\label{tab:a3-fib-pomdp}
	\end{center}
	\vspace{-0.1in}
\end{table*}

\subsection{Experimental Setup}
To evaluate the performance of the proposed POMDP algorithm, we  consider robotic navigation problems---\emph{cit, mit, pentagon sunysb, fourth}---available online\footnote{\label{footnote}https://www.pomdp.org/examples/}.
In each test scenario, the robot departs from the start state, and must reach the goal state.
The robot wants to maximize its total discounted reward while executing the policy given to reach the goal state. The rewards are specified as zero for all state-action pairs except when the robot declares itself to be in the goal state.
If the robot's state is also the goal state, it receives $1$ as a reward. Otherwise, $-1$ is given as a penalty. 
The robot has three actions that move it in one of the three possible directions, plus a fourth \emph{declare} action. The robot may fail to execute an action with some transition probability specified for each navigation problem. The robot detects four observations: a wall, a door, an open space, or \emph{undetermined} from neighboring states. Moreover, the observations are subject to sensor errors, which are modeled as the observation probability.
The transition and observation probabilities are given by each problem.

Two additional benchmark problems, namely \emph{Tag, Underwater}, are selected for performance evaluation.\footnote{https://bigbird.comp.nus.edu.sg/pmwiki/farm/appl/}
\emph{Underwater}  is a  navigation problem in which a  robot seeks to reach a goal state. However, unlike maze-like settings, \emph{Underwater} is a $51 \times 52$ grid map.
Moreover, it has a fixed set of available initial states distributed along the left side of the grid world.
In \emph{Underwater}, the robot can localize itself only in some specific regions. Thus, it requires the robot to reach such regions first and then re-plan to get to the destination.
 
Lastly, \emph{Tag}  is a problem of catching a moving target, so the environment is highly dynamic.
This makes \emph{Tag} significantly different from navigation problems, where the environments are static and known in advance.

For each of the test scenarios, we compare our AA-FIB algorithm with FIB and SARSOP~\cite{kurniawati2008sarsop}.
For a fair comparison, the  time limit for SARSOP is chosen as the minimum computation time required for AA-FIB.

The benchmark problems are solved with 100 different initial $\alpha_0$'s with all their components are randomly selected within $[{r_{\min}}/{(1-\gamma)},{r_{\max}}/{(1-\gamma)}]$.
After convergence, the discounted cumulative rewards are obtained by running 100 episodes of maximum trajectory length 100.
To compare the robustness and generality of the obtained policies, two types of initial beliefs are used when rolling out the policies: a fixed initial belief given for each problem and a randomly selected belief.
 All the experiments were conducted using Python 3.7.4 on a PC with Intel Core i7-8700K at 3.70GHz. The source code of our AA-FIB implementation is available online.\footnote{https://github.com/POMDP-core/AA-FIB}

\begin{table*} 
	\caption{Performance comparisons for the simulation-based AA-FIB algorithm on benchmark problems ($\mathrm{mean} \pm \mathrm{std}$).}
	\begin{center}
		\scalebox{0.89}{%
			\begin{tabular}{|l|l|c|c|c|c|c|}
				\hline
				& & \multicolumn{5}{c|}{\textbf{Algorithm}}  \\  \hline
				\multirow{2}{*}{\textbf{Problem ($|\mathcal{S}|, |\mathcal{A}|, |\mathcal{O}|$)} } & \multirow{2}{*}{\textbf{Metric}}  & \multirow{2}{*}{\textbf{FIB}} &  \multicolumn{4}{c|}{\textbf{Simulation-based AA-FIB (memory size $M_{\max}$)}}\\ 
				& & & $M_{\max} = 4$ & $M_{\max}=8$ & $M_{\max}=12$ & $M_{\max}=16$\\\hline\hline
				\multirow{4}{*}{cit ($284,4,28$)}
				&  $\#\mathrm{iter}$        & $1362.01 \pm 17.78$ & $\bf{536.34 \pm 75.48}$  & $708.46 \pm 135.67$ & $736.91 \pm 156.81$ & $576.04 \pm 161.51$ \\ 
				&  $t_{\mathrm{total}}$(sec)   & $4.036 \pm 0.057$   & $\bf{2.007 \pm 0.280}$   & $2.682 \pm 0.516$   & $2.669 \pm 0.553$   & $2.200 \pm 0.615$ \\ 
				&  $t_{\mathrm{AA}}$(sec)      & -                 & $0.049 \pm 0.007$   & $0.079 \pm 0.015$   & $0.085 \pm 0.018$   & $0.073 \pm 0.021$ \\ 
				&  $\mathrm{reward}_{\mathrm{rand}}$    & $0.44 \pm 0.05$     & $0.44 \pm 0.05$     & $0.43 \pm 0.05$     & $0.42 \pm 0.05$     & $0.44 \pm 0.04$\\
				&  $\mathrm{reward}_{\mathrm{fixed}}$    & $0.81 \pm 0.01$     & $0.82 \pm 0.01$     & $0.82 \pm 0.01$     & $0.81 \pm 0.02$     & $0.82 \pm 0.01$\\
				 \hline

				\multirow{4}{*}{mit ($204,4,28$)}
				&  $\#\mathrm{iter}$        & $1362.24 \pm 12.39$ & $426.22 \pm 80.09$  & $494.33 \pm 113.98$     & $422.30 \pm 121.91$     & $\bf{360.95 \pm 134.50}$ \\ 
				&  $t_{\mathrm{total}}$(sec)   & $4.021 \pm 0.054$   & $1.545 \pm 0.294$   &  $1.808 \pm 0.422$       &  $1.553 \pm 0.448$       & $\bf{1.330 \pm 0.500}$ \\ 
				&  $t_{\mathrm{AA}}$(sec)      & -                 & $0.039 \pm 0.008$   & $0.054 \pm 0.013$       & $0.050 \pm 0.014$       & $0.044 \pm 0.016$ \\ 
				&  $\mathrm{reward}_{\mathrm{rand}}$     & $0.65 \pm 0.04$     & $0.61 \pm 0.03$     & $0.61 \pm 0.04$         & $0.64 \pm 0.04$         & $0.64 \pm 0.04$\\
				&  $\mathrm{reward}_{\mathrm{fixed}}$    & $0.86 \pm 0.02$     & $0.87 \pm 0.01$     & $0.87 \pm 0.01$         & $0.87 \pm 0.01$         & $0.87 \pm 0.01$\\
 \hline
				
				\multirow{4}{*}{pentagon ($212,4,28$)}
				&  $\#\mathrm{iter}$        & $1361.22 \pm 17.87$ & $380.99 \pm 83.19$  & $458.72 \pm 113.56$ & $418.62 \pm 113.21$ & $\bf{361.29 \pm 129.22}$ \\ 
				&  $t_{\mathrm{total}}$(sec)   & $3.797 \pm 0.052$   & $1.396 \pm 0.304$   & $1.701 \pm 0.419$   & $1.557 \pm 0.424$   & $\bf{1.350 \pm 0.486}$ \\ 
				&  $t_{\mathrm{AA}}$(sec)     & -                 & $0.035 \pm 0.009$   & $0.051 \pm 0.01$    & $0.051 \pm 0.014$   & $0.045 \pm 0.016$ \\ 
				&  $\mathrm{reward}_{\mathrm{rand}}$     & $0.57 \pm 0.04$     & $0.58 \pm 0.05$     & $0.58 \pm 0.05$     & $0.57 \pm 0.05$     & $0.59 \pm 0.04$\\
				&  $\mathrm{reward}_{\mathrm{fixed}}$   & $0.83 \pm 0.01$     & $0.79 \pm 0.02$     & $0.80 \pm 0.01$     & $0.79 \pm 0.01$     & $0.80 \pm 0.01$\\
 \hline
				
				\multirow{4}{*}{sunysb ($300,4,28$)}
				&  $\#\mathrm{iter}$       & $1364.77 \pm 11.13$ & $\bf{532.36 \pm 85.47}$  & $738.55 \pm 175.03$ & $650.99 \pm 155.24$ & $574.13 \pm 170.16$ \\ 
				&  $t_{\mathrm{total}}$(sec)   & $4.280 \pm 0.071$   & $\bf{1.964 \pm 0.318}$   & $2.751 \pm 0.648$   & $2.440 \pm 0.580$   & $2.164 \pm 0.643$ \\ 
				&  $t_{\mathrm{AA}}$(sec)     & -                 & $0.048 \pm 0.008$   & $0.082 \pm 0.020$   & $0.076 \pm 0.018$   & $0.073 \pm 0.022$ \\ 
				&  $\mathrm{reward}_{\mathrm{rand}}$     & $0.42 \pm 0.04$     & $0.44 \pm 0.05$     & $0.42 \pm 0.04$     & $0.43 \pm 0.05$     & $0.43 \pm 0.05$\\
				& $\mathrm{reward}_{\mathrm{fixed}}$    & $0.76 \pm 0.02$     & $0.76 \pm 0.02$     & $0.76 \pm 0.01$     & $0.76 \pm 0.01$     & $0.77 \pm 0.01$\\
				 \hline
				
				\multirow{4}{*}{fourth ($1052,4,28$)}
				&  $\#\mathrm{iter}$        & $1364.02 \pm 14.57$ & $\bf{667.37 \pm 49.14}$  & $1018.84 \pm 132.36$    & $1197.66 \pm 211.04$    & $1249.08 \pm 270.66$ \\ 
				&  $t_{\mathrm{total}}$(sec)   & $6.241 \pm 0.092$   & $\bf{3.309 \pm 0.243}$   & $5.021 \pm 0.667$       & $5.982 \pm 1.057$       & $6.389 \pm 1.395$ \\ 
				&  $t_{\mathrm{AA}}$(sec)     & -                 & $0.082 \pm 0.006$   & $0.145 \pm 0.019$       & $0.200 \pm 0.035$       & $0.239 \pm 0.052$ \\ 
				&  $\mathrm{reward}_{\mathrm{rand}}$     & $0.39 \pm 0.03$     & $0.40 \pm 0.04$     & $0.39 \pm 0.03$         & $0.38 \pm 0.03$         & $0.40 \pm 0.04$\\
				&  $\mathrm{reward}_{\mathrm{fixed}}$    & $0.59 \pm 0.01$     & $0.59 \pm 0.01$     & $0.59 \pm 0.01$         & $0.60 \pm 0.01$         & $0.59 \pm 0.01$\\
				 \hline
				
				   \multirow{4}{*}{Tag ($870,5,30$)}
				&  $\#\mathrm{iter}$        & $315.52 \pm 0.52$ & $100.33 \pm 8.66$  & $95.85 \pm 11.04$    & $89.87 \pm 10.47$    & $\bf{83.44 \pm 5.75}$ \\ 
				&  $t_{\mathrm{total}}$(sec)   & $1.744 \pm 0.011$   & $0.614 \pm 0.064$   & $0.588 \pm 0.074$       & $0.554 \pm 0.069$       & $\bf{0.523 \pm 0.057}$ \\ 
				&  $t_{\mathrm{AA}}$(sec)     & -                 & $0.013 \pm 0.002$    & $0.013 \pm 0.002$    & $0.015 \pm 0.003$             & $0.016 \pm 0.001$ \\ 
				&  $\mathrm{reward}_{\mathrm{rand}}$     & $-15.89 \pm 0.88$     & $-15.54 \pm 0.99$     & $-15.31 \pm 0.91$         & $-15.34 \pm 0.82$         & $-15.37 \pm 0.95$\\
				&  $\mathrm{reward}_{\mathrm{fixed}}$    & $-17.35 \pm 0.70$     & $-16.49 \pm 0.76$     & $-16.23 \pm 0.63$         & $-16.46 \pm 0.67$         & $-16.46 \pm 0.78$\\
				 \hline
				 
				\multirow{4}{*}{Underwater ($2653,6,102$)}
				&  $\#\mathrm{iter}$        & $445.96 \pm 4.96$ & $\bf{144.90 \pm 6.14}$  & $182.14 \pm 7.96$    & $211.69 \pm 12.27$    & $233.22 \pm 14.88$ \\ 
				&  $t_{\mathrm{total}}$(sec)   & $16.965 \pm 0.344$   & $\bf{5.361 \pm 0.239}$   & $6.783 \pm 0.307$       & $7.956 \pm 0.487$       & $8.791 \pm 0.577$ \\ 
				&  $t_{\mathrm{AA}}$(sec)     & -                 & $0.031 \pm 0.003$   & $0.056 \pm 0.007$       & $0.089 \pm 0.011$       & $0.125 \pm 0.015$ \\ 
				&  $\mathrm{reward}_{\mathrm{rand}}$     & $3298.38 \pm 246.15$     & $3323.65 \pm 220.35$     & $3265.26 \pm 259.83$         & $3282.56 \pm 236.36$         & $3264.56 \pm 264.68$\\
				&  $\mathrm{reward}_{\mathrm{fixed}}$    & $-39.86 \pm 43.86$     & $-46.53 \pm 48.75$     & $-49.80 \pm 45.96$         & $-37.53 \pm 42.63$         & $-42.36 \pm 44.20$\\
				 \hline

		\end{tabular}
		}
		\label{tab:approx-a3-fib}
	\end{center}
			\vspace{-0.05in}
\end{table*}

\subsection{Model-Based Method}

Table \ref{tab:a3-fib-pomdp} shows the experiment results for the model-based version of our algorithm together with
FIB and SARSOP.
The computed policies are evaluated  according to the following metrics:
\begin{itemize}
	\item $\#\mathrm{iter}$: total number of iterations for convergence;
	\item $t_{\mathrm{total}}$(sec): total computation time for convergence;
	\item $t_{\mathrm{AA}}$(sec): total computation time for computing weights in the AA algorithm; 
	\item $\mathrm{reward}_{\mathrm{rand}}$: discounted cumulative reward obtained from randomly selected initial beliefs;
	\item $\mathrm{reward}_{\mathrm{fixed}}$: discounted cumulative reward obtained from a fixed initial belief.
\end{itemize}

As shown in Table \ref{tab:a3-fib-pomdp}, 
AA-FIB significantly reduces the total number of iterations and the total computation time in all benchmark problems compared to FIB.
This result indicates that AA improves the convergence speed of the FIB algorithm.
To compare the quality of the resulting policies, we also provide the results of the total discounted cumulative reward. The results show that the policy obtained by AA-FIB performs as well as that obtained by its standard counterpart. 
This observation confirms that AA does not degrade the quality of the solution.
We further examine the effects of  memory size in AA-FIB. 
For this purpose, we run the AA-FIB algorithm using different maximum memory sizes $\mmax = 4, 8, 12, 16$.\footnote{The regularization parameter was optimized for each memory size.}
The result indicates that  optimizing the memory size can notably speed up the convergence of our algorithm.
Overall, we can conclude that our method improves the convergence of FIB  without degrading the quality of the solution across all memory sizes on all benchmark problems.

Regarding the comparison with SARSOP, recall that its time limit is set as the minimum computation time for AA-FIB. 
Given the time limit, AA-FIB attains higher rewards compared to SARSOP on all the benchmark problems except \emph{Tag} and \emph{Underwater}.
 In particular, under randomly selected beliefs, 
 AA-FIB significantly outperforms SARSOP on the first five benchmark problems.

\begin{table*} 
	\caption{Effect of sample size on the performance of the simulation-based AA-FIB algorithm ($\mathrm{mean} \pm \mathrm{std}$).}
	\begin{center}
		\scalebox{1}{%
			\begin{tabular}{|l|c|c|c|c|c|c|}
				\hline
{\textbf{Sample size}} & {2} &{4} &  {6} &  {8} &  {10} &  {20} \\  \hline \hline

Error in policy ($\%$) & $14.81 \pm 23.86$ & $1.44 \pm 0.52$ & $1.20 \pm 0.42$ & $1.05 \pm 0.39$ & $0.98 \pm 0.35$ & $0.81 \pm 0.31$ 
\\ \hline

$\mathrm{reward}_{\mathrm{rand}}$ & $0.318 \pm 0.103$ & $0.396 \pm 0.055$ & $0.406 \pm 0.048$ &  $0.411 \pm 0.052$ & $0.420 \pm 0.050$ & $0.424 \pm 0.041$  
\\ \hline

$\mathrm{reward}_{\mathrm{fixed}}$ & $0.645 \pm 0.210$ & $0.724 \pm 0.042$ & $0.735 \pm 0.039$ &  $0.745 \pm 0.035$ & $0.734 \pm 0.037$ & 
$0.746 \pm 0.039$
\\ \hline

$\#\mathrm{iter}$ &  $611.88 \pm 147.23$ & $580.30 \pm 85.66$ & $539.42 \pm 100.19$ & $552.53 \pm 86.05$ & $540.31 \pm 81.29$ & $516.49 \pm 89.53$
\\ \hline

$t_{\mathrm{total}}$(sec) & $2.243 \pm 0.536$ & $2.148 \pm 0.318$ & $2.020 \pm 0.376$ & $2.079 \pm 0.324$ &  $2.040 \pm 0.307$ & $1.984 \pm 0.342$
\\ \hline

		\end{tabular}
		}
		\label{tab:samplesize}
	\end{center}
\end{table*}

\subsection{Simulation-Based Method}

Table~\ref{tab:approx-a3-fib} presents the experiment results for the simulation-based AA-FIB algorithm with a sample size of $| J | = 20$. 
As in the case of the exact version, 
our simulation-based method significantly outperforms the FIB algorithm in terms of the total number of iterations and computation time. 
Furthermore, the cumulative reward obtained by our simulation method is similar to that obtained by FIB. 
In fact, by comparing Tables~\ref{tab:a3-fib-pomdp} and \ref{tab:approx-a3-fib}, we observe that the performances of the exact version and the simulation version of our algorithm are comparable in terms of all the metrics, including the total computation time and the cumulative reward. 

Table.~\ref{tab:samplesize} shows the effect of sample size $| J |$.
First, the percentage error between the FIB solution $\alpha^\star$ and the solution obtained by our simulation method decreases with sample size. 
In particular, the error is very small even when $4$ sample data are used, and it is less than $1$\% with sample size no less than $10$.  
Likewise, 
the cumulative reward increases and then saturates with sample size bigger than $4$. 
As an interesting observation, the number of iterations decreases with sample size. We conjecture that this is because an accurate approximation of $F$ enhances the convergence speed of AA. 
As a result, the total computation time also decreases as more sample data are used.

\section{Conclusions and Future Work}

We proposed an accelerated POMDP algorithm carefully combining AA and FIB to attain their salient features. 
Our theoretical analyses showed the convergence of our algorithm to the FIB solution and
 identified a provable error bound for a simulation-based implementation. 
 The results of our experiments  confirm that the AA-FIB algorithm outperforms its standard counterpart in terms of both the total number of iterations and total computation time. Moreover, we confirm that the use of AA does not degrade the  quality of the solution. 
 The proposed method can be extended in several interesting ways such as $(i)$ improving scalability using neural networks as a function approximator, and $(ii)$ accelerating partially observable reinforcement learning.

\bibliographystyle{IEEEtran}
\bibliography{reference} 

\end{document}